\documentclass[11pt]{article}

\usepackage{amsmath}
\usepackage{amssymb}
\usepackage{amsthm} 
\usepackage{mathrsfs}
\usepackage[dvipdfmx]{graphicx,color}
\nopagebreak[3]

\newcommand{\newsection}[1]
{\section{#1}\setcounter{theorem}{0} \setcounter{equation}{0} \par\noindent}

\newtheorem{theorem}{Theorem}

\newtheorem{lemma}[theorem]{Lemma}

\newtheorem{remark}[theorem]{Remark}

\newcommand{\beq}{ \begin{equation} }
\newcommand{\eeq}{ \end{equation} }

\title{
On the derivation of several second order 
\\ 
partial differential equations from a generalization 
\\
of the Einstein equation}
\author
{
Makoto NAKAMURA
\thanks
{
{Faculty of Science, Yamagata University, 
Kojirakawa-machi 1-4-12, Yamagata 990-8560, JAPAN.}
E-mail: \texttt{nakamura@sci.kj.yamagata-u.ac.jp}
}
}
\date{}

\begin{document}

\maketitle

\begin{abstract}
A generalization of the Einstein equation is considered for complex line elements.
Several second order semilinear partial differential equations are derived from it 
as semilinear field equations in uniform and isotropic spaces.
The nonrelativistic limits of the field equations are also considered.
The roles of spatial variance are studied based on energy estimates, 
and several dissipative or antidissipative properties are remarked.
\end{abstract}

\noindent
{\it Mathematics Subject Classification (2010)}: Primary 35Q75; Secondary  35G20, 35Q76. \\

\noindent
{\it Keywords}: 
semilinear partial differential equations, 
nonrelativistic limit, 
Einstein equation.

%

\newsection{Introduction}
In this paper, we report some results on the relation between the Einstein equation 
and several second order partial differential equations.
We consider a generalization of the Einstein equation 
for non-Hermitian complex line elements of the form  
$g_{\alpha\beta}(z) dz^\alpha dz^\beta$, 
where $\{g_{\alpha\beta}\}_{0\le \alpha,\beta\le n}$ are complex-valued functions 
and $z=(z^0,\cdots, z^n)$ $\in \mathbb{C}^{1+n}$.
Under the cosmological principle, we give the solution of the generalized Einstein equation as   
\[
g_{\alpha\beta}dz^\alpha dz^\beta
=
-c^2(dz^0)^2+
a(z^0)^2q^2\left(1+\frac{k^2 r^2}{4} \right)^{-2}\sum_{j=1}^n (dz^j)^2,
\]
where $c>0$ is the speed of light, $q(\neq 0), k \in \mathbb{C}$ are constants, 
and 
$a$ is a complex-valued function which denotes the scale-function of the space.
This is known as the Robertson-Walker metric 
when $n=3$, $z=(t,x^1,x^2,x^3)\in \mathbb{R}^4$, $q=1$ and $k^2=0, \pm 1$,
where $k^2$ denotes the curvature of the space.
There is a large body of literature on the generalization of the Einstein equation  
for Hermitian line elements, complex line elements and general dimensions 
(see e.g., 
\cite{
ChoquetBruhat-2009-MilanJMath, 
ChoquetBruhat-2009-Oxford, 
Goenner-2004-LivingReviews,
Goenner-2014-LivingReviews}).

As the equation of motion of massive scalar field described  
by a complex-valued function $\phi=\phi(z^0,\cdots,z^n)$ 
with mass $m$ and potential $\lambda |\phi|^{p-1}\phi^2/(p+1)$ 
with $\lambda \in \mathbb{C}$ and $1\le p<\infty$, 
we derive the second order differential equation 
\begin{equation}
\label{Eq-Field-2-Intro}
-\frac{1}{c^2}\left(
\partial_0^2+\frac{n\partial_0 a}{a}\partial_0+\frac{m^2c^4}{\hbar^2}
\right)\phi
+
\frac{1}{a^{2}} \Delta_z \phi
+
\lambda   |\phi|^{p-1}\phi
=0,
\end{equation}
where $\hbar$ is the Planck constant, $\partial_0:=\partial/\partial z^0$ and 
$\Delta_z:=\sum_{j=1}^n \partial^2/(\partial z^j)^2$.
We also show the nonrelativistic limit of this equation yields the equation 
\begin{equation}
\label{Eq-Field-2-NR-Intro}
\pm i\frac{2m}{\hbar} \partial_0 u+\frac{1}{a^2} \Delta_z u
+
\lambda  |uw|^{p-1}u=0
\end{equation}
with a suitable transformation from $\phi$ to $u=u(z^0,\cdots,z^n)$, 
where $i:=(-1)^{1/2}$ 
and $w$ is a weight function defined by $w(z^0):=b_0(a(0)/a(z^0))^{n/2}$ 
for a constant $b_0\in \mathbb{C}$.
By a transformation (see \eqref{z-x}, below) from the complex coordinates $z$ to the real coordinates 
$(t,x^1,\cdots, x^n)\in \mathbb{R}^{1+n}$, 
the equations \eqref{Eq-Field-2-Intro} and \eqref{Eq-Field-2-NR-Intro} 
give typical second order partial differential equations.
For example, let us consider the simplest case that the scale-function is a constant 
$a=1$.
From \eqref{Eq-Field-2-Intro} and \eqref{Eq-Field-2-NR-Intro}, 
we obtain the Klein-Gordon equation 
\beq
\label{Intro-KG}
\partial_t^2 \phi
-c^2\Delta_x \phi
+\frac{m^2c^4}{\hbar^2} \phi
-c^2\lambda   |\phi|^{p-1}\phi
=0,
\eeq
the Schr\"odinger equation 
\beq
\label{Intro-Schroedinger}
\pm i\frac{2m}{\hbar} \partial_t u+ \Delta_x u
+
\lambda  |u|^{p-1}u=0,
\eeq
the elliptic equation
\beq
\label{Intro-Elliptic}
\partial_t^2 \phi
+
c^2\Delta_x \phi
+\frac{m^2c^4}{\hbar^2} \phi
-
c^2\lambda   |\phi|^{p-1}\phi
=0,
\eeq
and the parabolic equation 
\beq
\label{Intro-Heat}
\frac{2m}{\hbar} \partial_t u- \Delta_x u
-
i\lambda  |u|^{p-1}u=0
\eeq
(see Section \ref{Section-PDE}, below), 
where we have put $\Delta_x:=\sum_{j=1}^n \partial^2/(\partial x^j)^2$.
The terms $\lambda |\phi|^{p-1}\phi$ and $\lambda |u|^{p-1}u$ are fundamental semilinear terms 
in nonlinear theory to consider more complicated nonlinear terms. 
For the last parabolic equation, 
we note that the dimension of $\hbar/m$ in the SI units is 
$\textsf{M}^2\textsf{S}^{-1}$ 
(\textsf{M}: meter, \textsf{S}: second),  
which is equivalent to the dimension of the thermal diffusivity 
$K_1$ of the heat equation 
$\partial_t u-K_1 \Delta_x u=0$, 
and also to the dimension of the diffusion coefficient $K_2$ of the diffusion equation 
$\partial_t u- K_2 \Delta_x u=0$.

It is well-known that the Schr\"odinger equation 
\eqref{Intro-Schroedinger} is derived from the Klein-Gordon equation 
\eqref{Intro-KG} 
by the transform $\phi=u e^{\mp imc^2t/\hbar}$ and the nonrelativistic limit $c\rightarrow\infty$.
One of the aims of this paper is to show 
that the generalization of the Einstein equation gives a unified way 
to derive the above important partial differential equations.
It is remarkable that we have the unified derivation for the hyperbolic equation and the parabolic equation 
since these two equations have different physical backgrounds and derivations.
The solutions of the Klein-Gordon equation \eqref{Intro-KG} 
and the Schr\"odinger equation \eqref{Intro-Schroedinger} have the properties of waves, 
while the solutions of the equation \eqref{Intro-Heat} as the diffusion equation have the properties of particles. 
In complex coordinates, these equations  
are unified in the forms of \eqref{Eq-Field-2-Intro} and \eqref{Eq-Field-2-NR-Intro}. 
The properties of the equations are dependent on the local coordinates which we choose 
(see the wave-particle duality in 
\cite{deBroglie-1925-AnnDePhysique} and \cite{Einstein-1905-AnnalDerPhys}). 
In this sense, it is natural to consider the spacetime in complex coordinates, 
and we are able to regard \eqref{Eq-Field-2-Intro} and \eqref{Eq-Field-2-NR-Intro} 
as the equations which describe the properties of energy.


In Section \ref{Section-Einstein}, we consider a generalization of the Einstein equation.


In Section \ref{Section-Uniform}, 
we consider the spatial variance described by the scale-function 
$a$, which satisfies the Einstein equation with the cosmological constant 
in uniform and isotropic spaces.
The study of roles of the cosmological constant 
and the spatial variance is important to describe the history of the universe, 
especially, the inflation and the accelerating expansion of the universe 
(see, e.g., 
\cite{Guth-1981-PhysRevD}, 
\cite{Kazanas-1980-AstrophysicalJ},
\cite{Perlmutter-etal-1999-AstrophysicalJ},
\cite{Riess-Schmidt-etal-1998-AstronomicalJ},
\cite{Sato-1981-MonNotRAstrSoc}, 
\cite{Starobinsky-1980-PhysicsLettersB}).  
The scale-function \eqref{a}, below, follows from the equation of state when we regard the cosmological constant as the dark energy (see \eqref{Eq-State}, below).
In this paper, we study the cosmological constant from the point of view of 
partial differential equations.
Especially, we remark that some dissipative and antidissipative properties appear by the spatial variance.
These properties have been studied for the Klein-Gordon equation and the Schr\"odinger equation 
in de Sitter spacetime  
(see \cite{Nakamura-2014-JMAA} and \cite{Nakamura-2015-JDE} and the references therein).
We consider the properties for more general equations \eqref{Eq-Field-2-Intro} 
and \eqref{Eq-Field-2-NR-Intro} in this paper.

In Sections \ref{Section-Field} and \ref{Section-PDE}, we derive the above equations 
\eqref{Eq-Field-2-Intro}, \eqref{Eq-Field-2-NR-Intro}, 
\eqref{Intro-KG}, \eqref{Intro-Schroedinger}, \eqref{Intro-Elliptic} and  \eqref{Intro-Heat}.

In Section \ref{Section-Energy}, 
we consider the energy estimates of the equations 
\eqref{Eq-Field-2-Intro}  and \eqref{Eq-Field-2-NR-Intro}.
We show that the spatial variance described by the scale-function $a$ gives dissipative and antidissipative properties for the estimates.

In Section \ref{Section-Vilenkin}, 
we give some remarks on Vilenkin's model 
of the early universe in our setting of complex line elements 
(see \cite{Vilenkin-1982-PhysicsLetters} for Vilenkin's model).
Vilenkin's model shows that the purely imaginary time axis $it$ for $t\in \mathbb{R}$ plays an important role to describe the birth of the universe from nothing through the tunnel effect.
This fact is one of motivations to 
regard the axes of the spacetime as the lines in the complex plane 
(see \eqref{z-x}, below).

In Section \ref{Section-Geodesic}, 
we also consider the geodesic curves defined by the complex line elements, 
and we show that the conservation law of the Hamiltonian is dependent on the scale-function in local coordinates,
although it is independent of the scale-function in proper time.

%

One of motivations to consider the complex line element  
$g_{\alpha\beta} dz^\alpha dz^\beta$ is 
to generalize the following elementary observation.
The Riemann metric $(cdx^0)^2+(dx^1)^2+\cdots+(dx^n)^2$ 
and the Lorentz metric $-(cdx^0)^2+(dx^1)^2+\cdots+(dx^n)^2$ 
for $x=(x^0,x^1,\cdots, x^n)\in \mathbb{R}^{1+n}$ are unified in a single form 
$(cdz^0)^2+(dz^1)^2+\cdots+(dz^n)^2$ for $z\in \mathbb{C}^{1+n}$ 
since 
$(z^0,z^1,\cdots,z^n)=(x^0,x^1,\cdots,x^n)$ gives the former metric 
and  
$(z^0,z^1,\cdots,z^n)=(ix^0,x^1,\cdots,x^n)$ gives the latter metric.
Let us consider two coordinates 
$z=(z^0,z^1, \cdots, z^n)\in \mathbb{C}^{1+n}$ and 
$z_\ast=(z_\ast^0,z_\ast^1, \cdots, z_\ast^n)\in \mathbb{C}^{1+n}$ which satisfy the invariance of line elements 
\[
(c dz^0)^2+\sum_{j=1}^n (dz^j)^2
=
(c dz_\ast^0)^2+\sum_{j=1}^n (dz_\ast^j)^2.
\]
Let us assume 
$z^j=z_\ast^j$ for $2\le j\le n$ for simplicity.
For $\theta\in \mathbb{C}$, 
the transform  
\[
\begin{pmatrix}
cz_\ast^0 \\
z_\ast^1
\end{pmatrix}
= 
\begin{pmatrix}
\cos\theta & -\sin\theta \\
\sin\theta &\cos\theta
\end{pmatrix}
\begin{pmatrix}
cz^0 \\
z^1
\end{pmatrix}
\] 
satisfies this invariance.
For any fixed $-\pi/2<\omega\le \pi/2$, 
let us consider the lines 
$z^0=it$, $z^1=x^1e^{i\omega}$,
$z_\ast^0=it_\ast$ and $z_\ast^1=x_\ast^1 e^{i\omega}$ 
in the complex plane $\mathbb{C}$,
where $t,t_\ast, x^1,x_\ast^1\in \mathbb{R}$.
Then we have the transform 
\[
\begin{pmatrix}
ct_\ast \\
x_\ast^1
\end{pmatrix}
=
\begin{pmatrix}
\cos\theta & i e^{i\omega}\sin\theta \\
i e^{-i\omega} \sin\theta &\cos\theta
\end{pmatrix}
\begin{pmatrix}
ct \\
x^1
\end{pmatrix}.
\]
So that, 
if $\omega=\pi/2$ and $\theta  \in \mathbb{R}$, 
then we have the rotational transform 
\[
\begin{pmatrix}
ct_\ast \\
x_\ast^1
\end{pmatrix}
=
\begin{pmatrix}
\cos\theta & -\sin\theta \\
\sin\theta &\cos\theta
\end{pmatrix}
\begin{pmatrix}
ct \\
x^1
\end{pmatrix}.
\]
If $\omega=0$ and $i\theta \in \mathbb{R}$, 
then we have the Lorentz transform  
\[
\begin{pmatrix}
ct_\ast \\
x_\ast^1
\end{pmatrix}
=
\begin{pmatrix}
\cosh(i\theta) & \sinh(i \theta) \\
\sinh(i\theta) & \cosh(i\theta)
\end{pmatrix}
\begin{pmatrix}
ct \\
x^1
\end{pmatrix}.
\]
So that, the complex line element naturally unifies the rotational transform and the Lorentz transform.
Based on this observation, we further study the properties of the complex line elements 
through the Einstein equation.
We report some fundamental results (derivations and energy estimates) 
on the equations \eqref{Eq-Field-2-Intro} and \eqref{Eq-Field-2-NR-Intro} 
in this paper.
More detailed results on the equations will appear in the future.


\newsection{A generalization of the Einstein equation}
\label{Section-Einstein}
In this section, we generalize the Einstein equation  
for the case of three spatial dimensions and real line elements 
(see e.g., \cite{Carroll-2004-Addison} and \cite{DInverno-1992-Oxford})
into the case of general dimensions and complex line elements. 
Although we are based on the classical argument for the former case, 
we show the details for the completeness of the paper.
In the following, Greek letters $\alpha, \beta,\gamma,\cdots$ run from $0$ to $n$, 
Latin letters $j,k,\ell, \cdots$ run from $1$ to $n$.
We use the Einstein rule for the sum of indices of tensors, 
for example, 
${T^\alpha}_\alpha:=\sum_{\alpha=0}^n {T^\alpha}_\alpha$ 
and 
${T^i}_i:=\sum_{i=1}^n {T^i}_i$.
For any $C^1$-curve $C$ in the complex plane $\mathbb{C}$ 
connecting a point $A\in \mathbb{C}$ to a point $B\in \mathbb{C}$ 
parametrized by $z=z(x)\in \mathbb{C}$ for $x\in \mathbb{R}$ with $dz/dx\neq0$, 
we note that the integration by parts 
\[
\int_{C} \frac{df}{dz}(z)g(z) dz =(fg)(B)-(fg)(A)-\int_C f(z)\frac{dg}{dz}(z) dz
\]
holds for any $C^1$-functions $f$ and $g$.
For real variables $x=(x^0,\cdots,x^n)\in \mathbb{R}^{1+n}$  
and arbitrarily fixed real numbers $(\omega^0,\cdots,\omega^n)\in (-\pi/2, \pi/2]^{1+n}$,  
we consider complex variables $z=(z^0,\cdots,z^n)\in \mathbb{C}^{1+n}$ 
parametrized by 
\begin{equation}
\label{z-x}
z^\alpha=e^{i\omega^\alpha}x^\alpha.
\end{equation}
We put   
$\partial_\alpha:=\partial/\partial z^\alpha= e^{-i\omega^\alpha}\partial/\partial x^\alpha$.
We define a $(1+n)$-dimensional manifold  
$\mathcal{M}:=\{z\in \mathbb{C}^{1+n}\ |\ z^\alpha=e^{i\omega^\alpha}x^\alpha, \ x^\alpha\in \mathbb{R},\ 0\le \alpha\le n\}$.
We consider a bilinear symmetric complex-valued functional $\langle \cdot, \cdot \rangle$ 
on the vector space spanned by the vectors 
$\{\partial_\alpha\}_{0\le \alpha\le n}$. 
We put 
$g_{\alpha\beta}(z):=\langle \partial_\alpha, \partial_\beta\rangle$.
We denote by $(g_{\alpha\beta}(z))$ the matrix 
whose components are given by $\{g_{\alpha\beta}(z)\}_{0\le \alpha,\beta\le n}$.
Put $g(z) :={\rm det}(g_{\alpha\beta}(z))$. 
Let $(g^{\alpha\beta}(z))$ be the inverse matrix of 
$(g_{\alpha\beta}(z))$.
We consider a line element  
\begin{equation}
\label{Tau}
-(c d\tau)^2=(d\ell)^2:=g_{\alpha\beta}(z) dz^\alpha dz^\beta,
\end{equation}
where $\tau$ denotes the proper time 
and we take the square root of $(c d\tau)^2$ as 
$-\pi< \mbox{arg}\, (c d\tau) \le \pi$. 
We define $dz$ by  
\[
dz=dz^0\wedge \cdots \wedge dz^n
:=\sum_{\sigma} \mbox{\rm sgn} (\sigma) \, dz^{\sigma(0)} \cdots dz^{\sigma(n)},
\] 
where $\sigma$ denotes the permutation of $\{0,\cdots, n\}$.
For the change of variables $x$ to $y=(y^0,\cdots,y^n)\in \mathbb{R}^{1+n}$ by $y=y(x)$, 
we consider the complex variables $w=(w^0,\cdots,w^n)$ by $w^\alpha=e^{i\omega^\alpha} y^\alpha$.
Then we have 
$\mbox{det} (\partial z^\alpha/\partial w^\beta) \in \mathbb{R}$, 
$(-g(z))^{1/2}=|\mbox{det} (\partial w^\alpha/\partial z^\beta)| (-g(w))^{1/2}$, 
and 
$(-g(w))^{1/2} dw$ 
$=\mbox{sgn det} (\partial w^\alpha/\partial z^\beta)$ 
$(-g(z))^{1/2} dz$ 
by direct calculations, 
where  $g(w)$ denotes the determinant of 
$\left(g_{\alpha\beta}(w)\right)$ with 
$g_{\alpha\beta}(w):=\langle \partial/\partial w^\alpha, \partial/\partial w^\beta\rangle$,
and we take the square root of $-g$ as $-\pi< \mbox{arg} (-g)\le \pi$.
We have the fundamental results 
\begin{eqnarray*}
g_{\alpha\beta}\partial_\gamma g^{\alpha\beta}&=&-(\partial_\gamma g_{\alpha\beta})g^{\alpha\beta},\\
\partial_\gamma g^{\alpha\beta}&=&-g^{\alpha \mu}(\partial_\gamma g_{\mu\nu})g^{\nu\beta},\\
\partial_\gamma g&=&g g^{\alpha\beta}\partial_\gamma g_{\alpha\beta}
\end{eqnarray*}
by direct calculations.
For any contravariant tensor $T^\alpha$, 
we denote its parallel displacement from $z$ to $z+w$ by 
$\widetilde{T}^\alpha (z+w)
:=T^\alpha(z)-{\Gamma^\alpha}_{\beta\gamma} (z)T^\beta(z) w^\gamma$,
where ${\Gamma^\alpha}_{\beta\gamma}(z)$ 
denotes the proportional constant at $z$.
We assume the symmetry condition 
${\Gamma^\alpha}_{\beta\gamma}={\Gamma^\alpha}_{\gamma\beta}$,  
and 
\[
\left(g_{\alpha\beta}\widetilde{T}^\alpha\widetilde{T}^\beta\right)(z+w)
=
\left(g_{\alpha\beta}{T}^\alpha {T}^\beta\right)(z)
+O\Big(\sum_{0\le \alpha\le n} (w^\alpha)^2\Big)
\] 
for any $T^\alpha$ and $w^\alpha$. 
Then we have the Christoffel symbol
\begin{equation}
\label{Christoffel}
{\Gamma^\alpha}_{\beta\gamma}
=
\frac{1}{2}
g^{\alpha\delta}
\left(
\partial_\beta g_{\delta\gamma}+\partial_\gamma g_{\beta\delta}
-\partial_\delta g_{\beta\gamma}
\right).
\end{equation}
We define the covariant derivative $\nabla_\beta$ 
for $T^\alpha$ by 
\begin{equation*}
\nabla_\beta T^\alpha (z)
:=
\lim_{w^\beta\rightarrow0} 
\frac{T^\alpha(z+\widetilde{w^\beta})-\widetilde{T}^\alpha(z+\widetilde{w^\beta})}{w^\beta}
=\partial_\beta T^\alpha (z)+{\Gamma^\alpha}_{\beta\gamma} (z) T^\gamma (z),
\end{equation*}
where $\widetilde{w^\beta}:=(0,\cdots,0, w^\beta,0,\cdots,0)$ 
whose $\beta$-component is $w^\beta$ 
and the other components are $0$.
We note that 
$\nabla_\gamma g_{\alpha\beta}=0$ 
and 
$\nabla_\gamma g^{\alpha\beta}=0$ follow from 
\eqref{Christoffel}.
In general, we define  
\begin{multline*}
\nabla_\delta {T^{\alpha\beta \cdots} }_{\mu\nu\cdots}
:=
\partial_\delta {T^{\alpha\beta \cdots} }_{\mu\nu\cdots}
+
{\Gamma^\alpha}_{\delta\varepsilon} 
{T^{\varepsilon\beta \cdots} }_{\mu\nu\cdots}
+
{\Gamma^\beta}_{\delta\varepsilon} 
{T^{\alpha\varepsilon \cdots} }_{\mu\nu\cdots}
+\cdots
\\
-
{\Gamma^\varepsilon}_{\delta\mu} 
{T^{\alpha\beta \cdots} }_{\varepsilon\nu\cdots}
-
{\Gamma^\varepsilon}_{\delta\nu} 
{T^{\alpha\beta \cdots} }_{\mu\varepsilon\cdots}
-\cdots
\end{multline*}
for any tensor ${T^{\alpha\beta \cdots} }_{\mu\nu\cdots}$.
By direct calculations, we have 
\begin{eqnarray}
{\Gamma^\beta}_{\alpha\beta} 
&=& 
\partial_\alpha 
\left(
\log (-g)^{1/2}
\right), 
\\
\nabla_\alpha T^\alpha
&=&
\frac{1}{ (-g)^{1/2} } 
\partial_\beta 
\left( 
(-g)^{1/2} T^\beta
\right), 
\label{Eq-Nabla-T}
\\
\nabla_\alpha \nabla^\alpha \psi
&=&
\frac{1}{ (-g)^{1/2} } 
\partial_\beta 
\left( 
(-g)^{1/2} g^{\beta\gamma} \partial_\gamma \psi
\right)
\label{Eq-Div-Scalar}
\end{eqnarray}
for any tensor $T^\alpha$ and any scalar $\psi$.

We define the Riemann curvature tensor 
\[
{R^\delta}_{\alpha\beta\gamma}:=\partial_\beta {\Gamma^\delta}_{\alpha\gamma} 
-
\partial_\gamma {\Gamma^\delta}_{\alpha\beta}
+
{\Gamma^\delta}_{\varepsilon \beta} {\Gamma^\varepsilon}_{\alpha\gamma}
-
{\Gamma^\delta}_{\varepsilon\gamma} {\Gamma^\varepsilon}_{\alpha\beta}
\]
which is derived from ${R^\delta}_{\alpha\beta\gamma}T^\alpha
=(\nabla_\beta\nabla_\gamma-\nabla_\gamma\nabla_\beta)T^\delta$.
We define the Ricci tensor $R_{\alpha\beta}:={R^\gamma}_{\alpha\beta\gamma}$, 
and the scalar curvature $R:=g^{\alpha\beta}R_{\alpha\beta}$.
We define the Einstein tensor by 
$G_{\alpha\beta}:=R_{\alpha\beta}-g_{\alpha\beta} R/2$.
The change of upper and lower indices is done by $g_{\alpha\beta}$ and $g^{\alpha\beta}$, 
for example, 
${G^\alpha}_\beta:=g^{\alpha\gamma}G_{\gamma\beta}$.

Let $\Lambda\in \mathbb{C}$ be a constant, which is called the cosmological constant. 
Let us consider the variation by $g_{\alpha\beta}$ of the Einstein-Hilbert action 
$\int_{\mathcal{M}}  \left(R+2\Lambda\right) $ $(-g)^{1/2} \, dz$.
By the definitions of the Ricci tensor and the covariant derivative, 
and by the symmetry condition  ${\Gamma^\sigma}_{\nu\mu}={\Gamma^\sigma}_{\mu\nu}$, 
we have 
\[
\delta R_{\rho\mu} 
=
\nabla_\mu \left(\delta {\Gamma^\lambda}_{\rho \lambda}\right)
-
\nabla_\lambda \left(\delta {\Gamma^\lambda}_{\rho \mu}\right),
\]
where 
$\delta  {T^{\alpha\beta\cdots} }_{\mu\nu\cdots}$ 
denotes the variation of ${T^{\alpha\beta\cdots} }_{\mu\nu\cdots}$ by $g_{\alpha\beta}$.
Since we have 
$\delta R=(\delta g^{\alpha\beta})R_{\alpha\beta} +g^{\alpha\beta}\delta R_{\alpha\beta}$ 
and 
$
g^{\alpha\beta}\delta R_{\alpha\beta}
=
\nabla_\beta A^\beta$,
where we have put 
$A^\beta:=g^{\alpha\beta}\delta {\Gamma^\lambda}_{\alpha\lambda}
-
g^{\alpha\lambda}\delta {\Gamma^\beta}_{\alpha\lambda}$.
By \eqref{Eq-Nabla-T}, we have 
\[
\delta (R+2\Lambda)
=
(\delta g^{\alpha\beta} )R_{\alpha\beta}
+
\frac{1}{(-g)^{1/2} } \partial_\gamma \left( (-g)^{1/2} A^\gamma \right).
\]
Since we have $\delta (-g)^{1/2}
=
-(-g)^{1/2}g_{\alpha\beta} \left( \delta g^{\alpha\beta}\right)/2$, 
we obtain 
\[
\delta \int_{\mathcal{M}}  (R+2\Lambda) (-g)^{1/2} dz
=
\int_{\mathcal{M} }
\left(
G_{\alpha\beta}-\Lambda g_{\alpha\beta}
\right) 
(-g)^{1/2} \delta g^{\alpha\beta} dz
+
\int_{\mathcal{M}} \partial_\gamma 
\left(
(-g)^{1/2} A^\gamma
\right)
dz.
\]
Since the second term in the right hand side vanishes by the divergence theorem,
the Euler-Lagrange equation for the Einstein-Hilbert action 
is given by 
$G_{\alpha\beta}-\Lambda g_{\alpha\beta}=0$.

For a stress-energy tensor 
${T^\alpha}_\beta$,
we define the $(1+n)$-dimensional Einstein equation 
\begin{equation}
\label{Einstein}
{G^\alpha}_{\beta}-\Lambda {g^\alpha}_\beta
= \kappa \, \, {T^\alpha}_\beta, 
\end{equation}
where $\kappa$ is a constant and we assume that 
$\kappa c^4$ is independent of $c$.
For the case $n=3$ and real line elements, 
the constant $\kappa$ is called 
the Einstein gravitational constant which is given by 
$\kappa=8\pi \mathcal{G}/c^4$, 
where $\mathcal{G}$ is the Newton gravitational constant.
For the case $n\ge 3$ and complex line elements, 
we are able to generalize the constant $\kappa$ to 
\begin{equation}
\label{Kappa}
\kappa:=\frac{2(n-1)\pi^{n/2} \mathcal{G}}{(n-2)\Gamma(n/2) c^{4} },
\end{equation}
where $\Gamma$ denotes the gamma function.
Let us show the derivation of \eqref{Kappa}.
We denote the volume of the unit ball in $\mathbb{R}^n$ by 
$\Omega_n:=2\pi^{n/2}/n\Gamma(n/2)$.
We put $\hat{z}:=(z^1,\cdots,z^n)$, 
$r(\hat{z}):=\left\{\sum_{j=1}^n (z^j)^2\right\}^{1/2}$,
and $\omega^1=\cdots=\omega^n$ in \eqref{z-x}.
We define a function $E(\hat{z})$ by  
\[
E(\hat{z}):=
\left\{
\begin{array}{ll}
\displaystyle 
\frac{1}{(2-n)n\Omega_n} 
r(\hat{z})^{2-n} & \mbox{if}\ n\ge3, 
\\
\displaystyle 
\frac{1}{n\Omega_n} 
\log r(\hat{z}) & \mbox{if}\ n=2, 
\\
\displaystyle 
\frac{1}{n\Omega_n} 
r(\hat{z}) & \mbox{if}\ n=1.
\end{array}
\right.
\]
Since $E(\hat{x})$ for $\hat{x}=(x^1,\cdots,x^n)\in \mathbb{R}^n$ is the fundamental solution of the Laplacian, 
the function $E(\hat{z})$ satisfies 
\beq
\label{E-Delta}
\Delta_{\hat{z}} E(\hat{z})=\delta(\hat{z}),
\eeq 
where 
$\Delta_{\hat{z}}:=\sum_{j=1}^n \partial^2 /(\partial z^j)^2$ 
and 
$\delta$ denotes the Dirac $\delta$-function.
We assume that $(g_{\alpha\beta})$ is sufficiently close to the Minkowski matrix $(\eta_{\alpha\beta}):=\mbox{diag}(-c^2,1,\cdots,1)$.
Namely, we put $h_{\alpha\beta}:=g_{\alpha\beta}-\eta_{\alpha\beta}$, 
and we assume that $|h_{\alpha\beta}|$ is sufficiently small.
For a potential $\phi=\phi(\hat{z})$ and the Lagrangian 
$L(\hat{z}, d\hat{z}/d\tau):=\sum_{j=1}^n (dz^j/d\tau)^2/2-\phi(\hat{z})$,
the Euler-Lagrange equation for the action 
$\int L(\hat{z}, d\hat{z}/d\tau) d\tau$
is given by 
\begin{equation}
\label{Eq-Newton}
\frac{d^2 \hat{z}}{d\tau^2}+\nabla_{\hat{z}} \phi=0, 
\end{equation}
where we have put $\nabla_{\hat{z}}:=(\partial_1,\cdots,\partial_n)$.
We regard this equation as the equation of motion in our setting.
Since a natural extension of the Newton equation for a particle at $\hat{z}$ 
in the gravitational field by $K$-particles 
with mass $m(k)$ at $\hat{z}(k)$ for $1\le k\le K$ has the form 
\[
\frac{d^2 \hat{z} }{d\tau^2}
=
-
\sum_{k=1}^K \mathcal{G} \cdot \frac{m(k)}{ r( \hat{z}-\hat{z}(k) )^{n-1} }
\cdot
\frac{ \hat{z}-\hat{z}(k) }{ r(\hat{z}-\hat{z}(k)) },
\]
we formulate the Newton equation \eqref{Eq-Newton} by 
$\phi:=n\Omega_n \mathcal{G} \rho\ast_{\hat{z}} E$,
where $\rho=\rho(\hat{z})$ denotes the density of mass.
We note that 
\begin{equation}
\label{DeltaPhi-1}
\Delta_{\hat{z}} \phi=n\Omega_n \mathcal{G}\rho
\end{equation}
holds by \eqref{E-Delta}.
The Euler-Lagrange equation for the action
\[
\int \left( -g_{\alpha\beta} \frac{dz^\alpha}{d\tau} 
\frac{dz^\beta}{ d\tau }\right)^{1/2} d\tau
\]
yields the equation of the geodesic curve as  
\begin{equation}
\label{Geodesic}
\frac{d^2 z^\delta}{d\tau^2}+{\Gamma^\delta}_{\alpha\beta} 
\frac{d z^\alpha}{ d\tau } \frac{ dz^\beta }{ d\tau }=0.
\end{equation}
Let us assume 
$\partial_0 g_{\alpha\beta}\doteqdot 0$,
$g^{0k}\partial_k g_{00}\doteqdot 0$, 
$h^{\alpha\beta} \partial_\gamma h_{\delta\varepsilon}\doteqdot 0$
and 
$dz^j/dz^0\doteqdot0$.
Then we have 
${\Gamma^\lambda}_{00}\doteqdot -g^{\lambda k}\partial_k g_{00}/2$,
${\Gamma^0}_{00}\doteqdot 0$, 
${\Gamma^j}_{00}\doteqdot - \partial_j g_{00}/2$ 
and 
$d\tau\doteqdot dz^0$ 
by the definitions of ${\Gamma^\alpha}_{\beta\gamma}$ and $d\tau$.
So that, we have 
$d^2 z^j/(dz^0)^2 +{\Gamma^j}_{00} \doteqdot 0$ which yields 
\begin{equation}
\label{Eq-Phi-g00}
\partial_j \left( \phi+\frac{g_{00}}{2} \right)\doteqdot 0
\end{equation}
for $1\le j\le n$ by \eqref{Eq-Newton} and \eqref{Geodesic}.
Let us consider the case $\Lambda=0$ in \eqref{Einstein}.
We have 
\beq
\label{Kappa-RT}
(n-1)R=-2\kappa T
\eeq
and 
\begin{equation}
\label{Einstein-R}
(n-1)R_{\alpha\beta}=\kappa \left( (n-1)T_{\alpha\beta} - T g_{\alpha\beta} \right)
\end{equation}
when $n\ge1$ by \eqref{Einstein}.
Under the assumption 
$\partial_\alpha h_{\beta\gamma} \partial_\delta h_{\varepsilon \zeta}\doteqdot0$, 
we have 
\begin{equation}
\label{Weak-R00}
R_{00}
\doteqdot
-\partial_j {\Gamma^j}_{00}
\doteqdot
\frac{1}{2} \Delta_{\hat{z}} g_{00}\doteqdot-\Delta_{\hat{z}} \phi=-n\Omega_n \mathcal{G}\rho,
\end{equation}
where we have used the definition of Ricci tensor, 
the above fact ${\Gamma^j}_{00}\doteqdot-\partial_j g_{00}/2$, 
\eqref{Eq-Phi-g00} 
and 
\eqref{DeltaPhi-1}. 
We now consider the stress-energy tensor $T^{\alpha\beta}$ given by   
\[
T^{\alpha\beta}:=-\rho \, 
\frac{\partial z^\alpha}{\partial \tau} \, 
\frac{\partial z^\beta}{\partial \tau}
\] 
based on the analogy to the stress tensor of the perfect gas.
We have 
\beq
\label{Weak-T-0}
T^{\alpha\beta}\doteqdot 
\left\{
\begin{array}{ll}
-\rho & \mbox{if} \ (\alpha,\beta)=(0,0), \\
0 & \mbox{if} \ (\alpha,\beta)\neq(0,0)
\end{array}
\right.
\eeq
by $\partial z^j/\partial z^0\doteqdot 0$.
So that, we have 
\beq
\label{T-Rho}
T\doteqdot -\rho g_{00}
\eeq
and 
\begin{equation}
\label{Weak-T}
T_{\alpha\beta}\doteqdot 
\left\{
\begin{array}{ll}
-\rho(g_{00})^2 & \mbox{if} \ (\alpha,\beta)=(0,0), \\
0 & \mbox{if} \ (\alpha,\beta)\neq(0,0).
\end{array}
\right.
\end{equation}
Therefore, we obtain 
\beq
\label{Kappa-Rho}
n(n-1)\Omega_n \mathcal{G}\rho \doteqdot (n-2) \kappa \rho (g_{00})^2
\eeq
by \eqref{Einstein-R} and \eqref{Weak-R00}.
The required result \eqref{Kappa} holds when $n\ge3$ 
by $g_{00}\doteqdot -c^2$ and \eqref{Kappa-Rho}.
When $n=2$, we have 
$T^{\alpha\beta}\doteqdot0$ since $\rho=0$ by \eqref{Kappa-Rho}.
When $n=1$, we have 
$\kappa T^{\alpha\beta}\doteqdot0$ 
since we have $\kappa=0$ or $\rho=0$ by \eqref{Kappa-Rho}.

%

\newsection{Uniform and isotropic spaces}
\label{Section-Uniform}
We put $r:=\left(\sum_{j=1}^n (z^j)^2\right)^{1/2}$. 
We assume that 
the space is uniform and isotropic,
and we consider the line element  
\begin{equation}
\label{Line-Element}
g_{\alpha\beta}dz^\alpha dz^\beta:=-c^2(dz^0)^2+e^{h(z^0)} e^{f(r)} \sum_{j=1}^n (dz^j)^2, 
\end{equation}
where $h$ and $f$ are complex-valued functions.
This line element is uniform in the sense that 
for any two points $P$ and $Q$ in $\mathbb{C}^n$, the ratio of the coefficients 
$e^{h(z^0)} e^{f(r(P))} /e^{h(z^0)} e^{f(r(Q))}$ is independent of $z^0$.

By direct calculations, we have ${G^0}_j={G^j}_0=0$, 
\[
{G^0}_0:=\frac{n-1}{2c^2} 
\left\{
\frac{n}{4}(\partial_0 h)^2-c^2 e^{-h-f}
\left( f''+(n-1)\frac{f'}{r}+\frac{n-2}{4} (f')^2 \right) 
\right\},
\]
and 
\begin{multline*}
{G^j}_k:={g^j}_k 
\left\{
\frac{n-1}{2c^2}
\left(
\partial_0^2 h+\frac{n}{4}(\partial_0 h)^2
\right)
\right.
\\
\left.
-\frac{n-2}{2}e^{-h-f}
\left(
f''+(n-2)\frac{f'}{r}+\frac{n-3}{4} (f')^2 
\right)
\right\}
\\
+
\frac{n-2}{2}
e^{-h-f} 
\left(
f''-\frac{f'}{r}-\frac{(f')^2}{2}
\right)
\frac{z^jz^k}{r^2},
\end{multline*}
where $f':=d f/d r$.
Since the space is isotropic,  the coefficient of $z^j z^k$ must vanish.
So that, we assume that $f$ satisfies $f''-f'/r-(f')^2/2=0$, by which we obtain 
\begin{equation}
\label{f}
e^f=q^2\left(1+\frac{k^2 r^2}{4} \right)^{-2}
\end{equation}
for constants $q(\neq0), k \in \mathbb{C}$.
We define a function 
\begin{equation}
\label{a-h}
a(z^0):=e^{h(z^0)/2}.
\end{equation}
Let us consider the stress-energy tensor ${T^\alpha}_\beta$ of the perfect fluid 
\[
{T^\alpha}_\beta:=\mbox{diag} (\rho c^2,-p,\cdots,-p)
\]
for constant density $\rho$ and pressure $p$.
We put $\widetilde{\rho}:=\rho+\Lambda /\kappa c^2$ and 
$\widetilde{p}:=p-\Lambda/\kappa$.
Then \eqref{Einstein} is rewritten as  
${G^\alpha}_\beta
=
\kappa \cdot \mbox{diag}(\widetilde{\rho}c^2, -\widetilde{p}, \cdots, -\widetilde{p})$, 
which shows that the cosmological constant $\Lambda>0$ is regarded 
as the energy which has positive density and negative pressure in the vacuum 
$\rho=p=0$ for $\kappa>0$ 
(``the dark energy" for $n=3$). 
The equation ${G^0}_{0}=\kappa \widetilde{\rho} c^2{g^0}_0$ is rewritten as  
\begin{equation}
\label{G00}
\frac{n-1}{2}
\left\{
\left(
\frac{\partial_0 a}{c a}
\right)^2
+
\frac{k^2}{q^2a^2}
\right\}
=\frac{\kappa c^2}{n} \cdot \widetilde{\rho}. 
\end{equation}
The equation ${G^j}_{k}=-\kappa \widetilde{p} {g^j}_k$ is rewritten as  
\begin{equation}
\label{Gjk}
\frac{n-1}{2}
\left\{
\frac{2}{n-2} \cdot \frac{\partial_0^2a}{c^2 a}
+
\left(
\frac{\partial_0 a}{c a}
\right)^2
+
\frac{k^2}{q^2a^2}
\right\}
=-\frac{\kappa}{n-2}\cdot  \widetilde{p},
\end{equation}
which is rewritten as the Raychaudhuri equation 
\begin{equation}
\label{Ray}
\frac{\partial_0^2 a}{c^2 a}
=-\frac{n-2}{n-1}\cdot \kappa 
\left(
\frac{\widetilde{\rho}c^2}{n}
+
\frac{\widetilde{p}}{n-2}
\right)
\end{equation}
by \eqref{G00}.
Multiplying $a^n$ to the both sides in \eqref{G00}, 
taking the derivative by $z^0$ variable, 
and using \eqref{Gjk}, 
we have the conservation of mass  
\begin{equation}
\label{Conservation-Mass}
\partial_0(\widetilde{\rho} c^2 a^n)+\widetilde{p} \partial_0 a^n=0.
\end{equation} 
For any number $\sigma$, we assume the equation of state 
\begin{equation}
\label{Eq-State}
\widetilde{p}=\sigma \widetilde{\rho} c^2.
\end{equation}
Then $a(z^0)$ must satisfy 
\[
\frac{\partial_0^2 a(z^0)}{c^2 a(z^0)}
=
-\frac{n-2+n\sigma}{n(n-1)} \cdot \kappa \widetilde{\rho} c^2 
\]
with 
\beq
\label{Tilrho}
\widetilde{\rho}=\frac{n-1}{2}\cdot \frac{n}{\kappa c^4}
\cdot \frac{\partial_0 a(0)^2}{a(0)^{2-n(1+\sigma)} } 
\cdot a(z^0)^{-n(1+\sigma)}
\eeq
by \eqref{Ray} and \eqref{Conservation-Mass}, which has the solution 
\begin{equation}
\label{a}
a(z^0):= 
\left\{
\begin{array}{ll}
a(0)\left( 1+\frac{n(1+\sigma)\partial_0a(0)z^0}{2a(0)} \right)^{2/n(1+\sigma)}
& 
\mbox{if}\ \ \sigma \neq -1, 
\\
a(0)\exp \left( \frac{\partial_0 a(0) z^0}{a(0)} \right) & 
\mbox{if}\ \ \sigma = -1.
\end{array}
\right.
\end{equation}
By the above argument, we have derived the line element 
\begin{equation}
\label{Line-Element-a}
g_{\alpha\beta}dz^\alpha dz^\beta
=
-c^2(dz^0)^2+
a(z^0)^2q^2\left(1+\frac{k^2 r^2}{4} \right)^{-2}\sum_{j=1}^n (dz^j)^2
\end{equation}
for constants $q(\neq 0), k \in \mathbb{C}$ 
as the solution of \eqref{Einstein}.
By \eqref{G00}, \eqref{Tilrho} and \eqref{a}, we have $k=0$.


\begin{remark}
The line element \eqref{Line-Element-a} is known as the Robertson-Walker metric 
for the case that $a(>0)$ is real-valued,  $z\in \mathbb{R}^{1+n}$, $n=3$, $q=1$ and $k^2=0, \pm 1$. 
Here, $k^2$ denotes the curvature of the space.
In this case, $a$ in \eqref{a} blows up in finite time 
when $\partial_0a(0)>0$ and $\sigma<-1$, 
which is called Big-Rip in cosmology. 
The case $\sigma= -1$ shows the exponential expansion 
of $a$ when $\partial_0a(0)>0$.
The case $\sigma> -1$ shows the polynomial expansion 
of $a$ when $\partial_0a(0)>0$.
These models are studied for the accelerating expansion  
of the universe.
The line element \eqref{Line-Element-a} with \eqref{a} is a natural extension of these models for general dimensions and complex line elements.
\end{remark}


\newsection{A field equation and the nonrelativistic limit}
\label{Section-Field}
For any $\lambda  \in \mathbb{C}$ and any complex-valued $C^2$ function $\phi$ on $\mathcal{M}$, 
we define the Lagrangian 
\[
L(\phi):=-\frac{1}{2} g^{\alpha\beta} \partial_\alpha \phi \, \partial_\beta \phi
-
\frac{1}{2}\left(\frac{mc}{\hbar}\right)^2\phi^2
+
\frac{\lambda  }{p+1}|\phi|^{p-1}\phi^2.
\]
We apply the variational method to the action  
$\int_{ \mathcal{M} } L(\phi) (-g)^{1/2} dz$ for $\phi$. 
Then the Euler-Lagrange equation is given by 
\begin{equation}
\label{Eq-Field}
\frac{1}{ (-g)^{1/2} } \partial_\alpha 
((-g)^{1/2} g^{\alpha\beta} \partial_\beta \phi)
-
\left(\frac{mc}{\hbar}\right)^2\phi 
+
\lambda   |\phi|^{p-1}\phi 
=0
\end{equation}
under the constraint condition $\mbox{arg}\, \delta \phi= \mbox{arg}\, \phi$.
This is the equation of motion of massive scalar field described by a function $\phi$ with the mass $m$ and the potential $\lambda   |\phi|^{p-1}\phi^2/(p+1)$.

We put $q=1$ and $k=0$.
Then the line element \eqref{Line-Element-a} is rewritten as 
$-c^2(dz^0)^2+a(z^0)^2 \sum_{1\le j\le n} (dz^j)^2$.
Then the field equation \eqref{Eq-Field} is rewritten as 
\begin{equation}
\label{Eq-Field-2}
-\frac{1}{c^2}\left(
\partial_0^2+\frac{n\partial_0 a}{a}\partial_0+\frac{m^2c^4}{\hbar^2}
\right)\phi
+
\frac{1}{a^{2}} \Delta_z \phi
+
\lambda   |\phi|^{p-1}\phi
=0.
\end{equation}
For any constant $b_0\in \mathbb{C}$, we define a weight function $w(z^0)$ and a function $b(z^0)$ by 
\[
w(z^0):=b_0 \left( \frac{a(0)}{a(z^0)} \right)^{n/2},
\ \ \ \ 
b(z^0):= w(z^0) \exp\left(\mp i \frac{m}{\hbar}c^2 z^0 \right),
\]
where we note $b(0)=b_0$.
We transform $\phi$ to $u$ by the equation 
\[
\phi(z^0,\cdots,z^n)=u(z^0,\cdots,z^n) b(z^0).
\]
We assume $mz^0/\hbar\in \mathbb{R}$.
Then the nonrelativistic limit ($c\rightarrow\infty$) of \eqref{Eq-Field-2} yields 
\begin{equation}
\label{Eq-Field-2-NR}
\pm i\frac{2m}{\hbar} \partial_0 u+\frac{1}{a^2} \Delta_z u
+
\lambda  |uw|^{p-1}u=0,
\end{equation}
where we have used the gauge invariance of $\lambda |\phi|^{p-1}\phi$.

%

\newsection{A unified derivation of several PDEs}
\label{Section-PDE}
Let us consider the equations \eqref{Eq-Field-2} and \eqref{Eq-Field-2-NR} 
under the transform \eqref{z-x} with 
$\omega^1=\cdots=\omega^n$.
We put $t=x^0$ and 
\beq
\label{Def-*}
\begin{array}{l}
\phi_\ast(t,x^1,\cdots,x^n):=\phi(z^0, \cdots,z^n), \\
u_\ast(t,x^1,\cdots,x^n):=u(z^0, \cdots,z^n), \\
a_\ast(t):= a(z^0),\ \ w_\ast(t):=w(z^0).
\end{array}
\eeq
We put $\theta:=\mbox{arg} \,\left( a_\ast(t)\right)$.
Then \eqref{Eq-Field-2} is rewritten as 
\begin{multline}
\label{Eq-Field-2-V-*}
-\frac{1}{c^2}
\frac{e^{2i(\theta+\omega^1)} }{e^{2i\omega^0} }
\left(
\partial_t^2+\frac{n\partial_t a_\ast}{a_\ast}\partial_t+\left(\frac{mc^2e^{i\omega^0}}{\hbar}\right)^2
\right)\phi_\ast
+
\frac{1}{|a_\ast|^{2}} \Delta_x \phi_\ast
\\
+
e^{2i(\theta+\omega^1)}\lambda |\phi_\ast|^{p-1}\phi_\ast
=0, 
\end{multline}
and  \eqref{Eq-Field-2-NR} is rewritten as 
\begin{equation}
\label{Eq-Field-2-NR-V-*}
\pm i\frac{2me^{i\omega^0}}{\hbar} 
\partial_t u_\ast
+
\frac{e^{2i\omega^0} }{e^{2i(\theta+\omega^1)} }
\left(
\frac{1}{|a_\ast|^2} \Delta_x u_\ast
+
e^{2i(\theta+\omega^1)} \lambda |u_\ast w_\ast|^{p-1}u_\ast
\right)
=0.
\end{equation}
The equation \eqref{Eq-Field-2-V-*} 
and its nonrelativistic limit \eqref{Eq-Field-2-NR-V-*} give 
a unified derivation of the elliptic equation, 
the Klein-Gordon equation, the Schr\"odinger equation, 
and the parabolic equation as follows.
For simplicity, we consider the simplest case $a=1$ 
which follows from $a(0)=1$ and $\partial_0a(0)=0$ in 
\eqref{a}.
So that,  we have $\theta=0$ and $a_\ast =1$.
We put $b_0=1$, which yields $w_\ast=1$.
When $\omega^0=\cdots=\omega^n=0$, 
\eqref{Eq-Field-2-V-*} and \eqref{Eq-Field-2-NR-V-*}
are rewritten as 
the Klein-Gordon equation 
\begin{equation}
\label{Eq-Field-2-Rewritten}
\partial_t^2 \phi_\ast
+\frac{m^2c^4}{\hbar^2} \phi_\ast
-c^2
\Delta_x \phi_\ast
-c^2
\lambda   |\phi_\ast|^{p-1}\phi_\ast
=0,
\end{equation}
and the Schr\"odinger equation  
\begin{equation}
\label{Eq-Field-2-NR-Rewrriten-Schrodinger}
\pm i\frac{2m}{\hbar} \partial_t u_\ast+ \Delta_x u_\ast
+
\lambda  |u_\ast|^{p-1}u_\ast=0,
\end{equation}
respectively.
When $\omega^0=0$ and $\omega^1=\cdots=\omega^n=\pi/2$, 
\eqref{Eq-Field-2-V-*} is rewritten as 
the elliptic equation 
\begin{equation}
\label{Eq-Field-2-Rewritten}
\partial_t^2 \phi_\ast
+\frac{m^2c^4}{\hbar^2} \phi_\ast
+
c^2\Delta_x \phi_\ast
-
c^2\lambda   |\phi_\ast|^{p-1}\phi_\ast
=0.
\end{equation}
When $\omega^0=0$ and $\omega^1=\cdots=\omega^n=\pi/4$, 
\eqref{Eq-Field-2-NR-V-*} with the positive sign is rewritten as 
the parabolic equation 
\begin{equation}
\label{Eq-Field-2-NR-Rewrriten-Parabolic}
\frac{2m}{\hbar} \partial_t u_\ast- \Delta_x u_\ast
-
i\lambda  |u_\ast|^{p-1}u_\ast=0.
\end{equation}

We note that \eqref{Eq-Field-2-V-*} is rewritten as 
\begin{equation}
\label{Eq-Field-2-Rewritten-a}
\partial_t^2 \phi_\ast
+nH\partial_t \phi_\ast
+\frac{m^2c^4}{\hbar^2} \phi_\ast
-\frac{c^2}{e^{2Ht} }
\Delta_x \phi_\ast
-c^2
\lambda   |\phi_\ast|^{p-1}\phi_\ast
=0
\end{equation}
when $\omega^0=\cdots=\omega^n=0$ and $a_\ast(t)=e^{Ht}$ (de Sitter spacetime) with the Hubble constant $H\in \mathbb{R}$.
So that, we know that the spatial expansion $H>0$ yields the effect of dissipation, 
while the spatial contraction $H<0$ yields the effect of antidissipation 
(see e.g., \cite{Nakamura-2014-JMAA} for the Klein-Gordon equation in de Sitter spacetime).
We also note that the complex Ginzburg-Landau equation 
\begin{equation}
\label{Eq-CGL}
\partial_t u_\ast
-
\gamma\Delta u_\ast
-
\lambda_1  u_\ast
+
\lambda_2  |u_\ast|^2u_\ast=0,  
\end{equation}
where $\gamma\in \mathbb{C}$ with $\mbox{\rm Re}\, \gamma>0$, 
$\lambda_1\ge0$,  
$\lambda_2\in \mathbb{C}$ with $\mbox{\rm Re}\, \lambda_2>0$, 
is considered as the sum of the potentials with  
$p=1$ and $p=3$ in \eqref{Eq-Field-2-NR-V-*} 
when $\gamma=\pm i\hbar/2me^{2i\omega^1}$, $\theta=\omega^0=0$ and $a_\ast=1$.

%

\newsection{Energy estimates}
\label{Section-Energy}
We have derived the equations \eqref{Eq-Field-2} and \eqref{Eq-Field-2-NR}.
The properties of the equations depend on the coordinate $(t,x)$ defined by $t:=x^0$ and \eqref{z-x}.
For example, we will see that the equations have dissipative or antidissipative properties 
on energy estimates dependently on $\omega^0,\cdots,\omega^n$.
We assume that there exist two functions $V_0$ and $V_0'$ on $\mathbb{C}$ which satisfy  
\begin{equation}
\label{Assumption-V0dt}
\partial_t V_0(\psi)=\mbox{\rm Re}\, \left( \partial_t \overline{\psi} \, V_0'(\psi) \right)
\end{equation}
for any function $\psi=\psi(t,x)$.
We put $\omega^1=\cdots=\omega^n$, $t=x^0$ and \eqref{Def-*}.
We assume that $\theta:=\mbox{arg} \,\left( a_\ast(t)\right)$ is a constant.
We put $V':=-e^{-2i(\theta+\omega^1)} V_0'$.
Let us consider energy estimates for the equations 
\begin{equation}
\label{Eq-Field-2-V}
-\frac{1}{c^2}\left(
\partial_0^2+\frac{n\partial_0 a}{a}\partial_0+\frac{m^2c^4}{\hbar^2}
\right)\phi
+
\frac{1}{a^{2}} \Delta_z \phi
+
V'(\phi)
=0
\end{equation}
and 
\begin{equation}
\label{Eq-Field-2-NR-V}
\pm i\frac{2m}{\hbar} \partial_0 u+\frac{1}{a^2} \Delta_z u
+
\frac{1}{w}V'(uw)=0,
\end{equation}
which are extensions of \eqref{Eq-Field-2} and \eqref{Eq-Field-2-NR} for general nonlinear terms $V'$.
When $V'(\phi)=\lambda  |\phi|^{p-1}\phi$, we have \eqref{Eq-Field-2} and \eqref{Eq-Field-2-NR}.
For example, $V_0(\psi):=\lambda_0  |\psi|^{p+1}/(p+1)$ and $V_0'(\psi):=\lambda_0  |\psi|^{p-1}\psi$ for  
$\lambda_0 \in \mathbb{R}$ satisfy 
\eqref{Assumption-V0dt}.
We use \eqref{Def-*}.
The equations \eqref{Eq-Field-2-V} and \eqref{Eq-Field-2-NR-V} are rewritten as  
\begin{multline}
\label{Eq-Field-2-V-*-V}
-\frac{1}{c^2}
\frac{e^{2i(\theta+\omega^1)} }{e^{2i\omega^0} }
\left(
\partial_t^2+\frac{n\partial_t a_\ast}{a_\ast}\partial_t+\left(\frac{mc^2e^{i\omega^0}}{\hbar}\right)^2
\right)\phi_\ast
\\
+
\frac{1}{|a_\ast|^{2}} \Delta_x \phi_\ast
-V_0'(\phi_\ast)
=0, 
\end{multline}
and  
\begin{equation}
\label{Eq-Field-2-NR-V-*-V}
\pm i\frac{2me^{i\omega^0}}{\hbar} 
\partial_t u_\ast
+
\frac{e^{2i\omega^0} }{e^{2i(\theta+\omega^1)} }
\left(
\frac{1}{|a_\ast|^2} \Delta_x u_\ast
-
\frac{1}{w_\ast} V_0'(u_\ast w_\ast)
\right)
=0.
\end{equation}

We put $C_0:=2me^{i\omega^0}/\hbar$.
We have the energy estimate for \eqref{Eq-Field-2-V-*-V} as follows.

\begin{lemma}
\label{Lemma-Conservation-KG}
 (Energy estimates.) 
Let us consider \eqref{Eq-Field-2-V-*-V}.
Assume $C_0\in \mathbb{R}$, 
$e^{2i(\theta+\omega^1)}/e^{2i\omega^0}$ $\in \mathbb{R}$ and \eqref{Assumption-V0dt}.
Then we have 
\beq
\label{Energy-KG}
\int_{\mathbb{R}^n} e^0(t,x) dx +\int_0^t \int_{\mathbb{R}^n} e^{n+1}(s,x) dx ds 
=
\int_{\mathbb{R}^n} e^0(0,x) dx,
\eeq
where we have put 
\[
e^0:=
\frac{ e^{ 2i(\theta+\omega^1) } }{ c^2e^{2i\omega^0} }
\left(
|\partial_t \phi_\ast|^2
+
\frac{C_0^2c^4}{4}|\phi_\ast|^2
\right)
+
\frac{1}{ |a_\ast|^2 } 
\sum_{j=1}^n |\partial_{x^j} \phi_\ast|^2
+
2V_0(\phi_\ast)
\]
and 
\[
e^{n+1}:=
\frac{ e^{ 2i(\theta+\omega^1)} }{ c^2e^{2i\omega^0} }
 2 \mbox{\rm Re}\,  
\left(
\frac{n\partial_t a_\ast}{a_\ast} 
\right)
|\partial_t \phi_\ast|^2
-
\partial_t 
\left(
\frac{1}{|a_\ast|^2}
\right)
\sum_{j=1}^n |\partial_{x^j} \phi_\ast|^2.
\]
\end{lemma}

\begin{proof}
We put 
\[
e^j:=-\frac{1}{ |a_\ast|^2 } 2 \mbox{\rm Re}\, 
\left(
\partial_t \overline{\phi_\ast} \partial_{x^j} \phi_\ast
\right)
\]
for $1\le j\le n$.
Multiplying $\partial_t\overline{\phi_\ast}$ to the both sides in \eqref{Eq-Field-2-V-*-V} 
and taking its real part, we have
\[
\partial_t e^0+\sum_{j=1}^n \partial_{x^j} e^j+e^{n+1}=0,
\]
where we have used $C_0\in \mathbb{R}$, $e^{2i(\theta+\omega^1)}/e^{2i\omega^0}\in \mathbb{R}$ 
and  
\eqref{Assumption-V0dt}.
The required result follows from the integration for $t$ and $x$.
\end{proof}

We have charge and energy estimates for \eqref{Eq-Field-2-NR-V-*-V} as follows.

\begin{lemma}
\label{Lemma-Conservation-CGL}
Let us consider \eqref{Eq-Field-2-NR-V-*-V}.
Assume $C_0\in \mathbb{R}$.
Let $V'_0$ satisfy 
\beq
\label{Assumption-zV0}
\mathrm{Im}\{\bar{z} V_0'(z)\}=0
\eeq
for any $z$ and its complex conjugate $\bar{z}$.

(1) (Charge estimates.)  
We have 
\beq
\label{Energy-C}
\pm \int_{\mathbb{R}^n} e_C^0(t,x) dx +\int_0^t \int_{\mathbb{R}^n} e_C^{n+1}(s,x) dx ds 
=
\pm \int_{\mathbb{R}^n} e_C^0(0,x) dx,
\eeq
where we have put 
\[
e_C^0:=C_0 |u_\ast|^2
\]
and 
\[
e_C^{n+1}:=2\mbox{\rm Im}\, 
\left(
\frac{ e^{2i(\theta+\omega^1)} }{ e^{2i\omega^0} }
\right)
\left(
\frac{1}{ |a_\ast|^2 } \sum_{j=1}^n |\partial_{x^j} u_\ast |^2
+
\frac{1}{ |w_\ast|^2 } \overline{u_\ast w_\ast} V_0'(u_\ast w_\ast)
\right).
\]

(2) (Energy estimates.)  
Assume \eqref{Assumption-V0dt}.
Then we have 
\beq
\label{Energy-E}
\int_{\mathbb{R}^n} e_E^0(t,x) dx +\int_0^t \int_{\mathbb{R}^n} e_E^{n+1}(s,x) dx ds 
=
\int_{\mathbb{R}^n} e_E^0(0,x) dx,
\eeq
where we have put 
\[
e_E^0:=
\sum_{j=1}^n |\partial_{x^j} u_\ast |^2
+
\frac{2|a_\ast|^2}{|w_\ast|^2} V_0(u_\ast w_\ast)
\]
and 
\begin{multline*}
e_E^{n+1}:=
\pm 2C_0\mbox{\rm Im}\, 
\left(
\frac{ e^{2i(\theta+\omega^1)} }{ e^{2i\omega^0} }
\right)
|a_\ast|^2 |\partial_t u_\ast|^2
\\
+\frac{n}{2|w_\ast|^2} \left(\partial_t |a_\ast|^2\right)
\cdot
\left(
\overline{u_\ast w_\ast} V_0'(u_\ast w_\ast)-\frac{2(n+2)}{n} V_0(u_\ast w_\ast)
\right).
\end{multline*}
\end{lemma}

\begin{proof}
(1) Multiplying $\overline{u_\ast}$ to the both sides in \eqref{Eq-Field-2-NR-V-*-V} 
and taking its imaginary part, we have 
\[
\pm \partial_t e_C^0+\sum_{j=1}^n \partial_{x^j} e_C^j+e_C^{n+1}=0,
\]
where we have used $C_0\in \mathbb{R}$, 
\eqref{Assumption-zV0}, and we have put 
\[ 
e_C^j:=2\mbox{\rm Im}\, 
\left(
\frac{ e^{2i\omega^0} }{ |a_\ast|^2 e^{2i(\theta+\omega^1)} }
\overline{u_\ast} \partial_{x^j} u_\ast
\right)
\]
for $1\le j\le n$.
The required result follows from the integration for $t$ and $x$.

(2) We rewrite \eqref{Eq-Field-2-NR-V-*-V} as 
\[
\pm i C_0 e^{2i(\theta+\omega^1-\omega^0)} |a_\ast|^2 \partial_t u_\ast 
+
\sum_{j=1}^n \partial_{x^j}^2 u_\ast
-
\frac{|a_\ast|^2}{w_\ast} V_0'(u_\ast w_\ast)=0.
\]
Multiplying $\partial_t\overline{u_\ast}$ to the both sides in this equation and taking its real part, we have
\begin{equation}
\label{Proof-Energy-NR-E}
\mp C_0 2\mbox{Im}\, \frac{ e^{2i(\theta+\omega^1)} }{ e^{2i\omega_0} }
|a_\ast|^2|\partial_t u_\ast|^2
-
\sum_{j=1}^n \partial_{x^j} e_E^j
-
\partial_t \sum_{j=1}^n |\partial_{x^j} u_\ast|^2
-
2\mbox{Re}\, I=0,
\end{equation}
where we have put 
$
e_E^j:=-2 \mbox{\rm Re} 
\left(
\partial_t \overline{u_\ast} \partial_{x^j} u_\ast
\right)
$ 
for $1\le j\le n$ and 
\[
I:=\frac{|a_\ast|^2}{w_\ast} \partial_t \overline{u_\ast} V_0'(u_\ast w_\ast).
\] 
Since we have 
\[
I
=
\frac{|a_\ast|^2}{|w_\ast|^2} 
\left(
\partial_t (\overline{u_\ast w_\ast}) V_0'(u_\ast w_\ast)-
\frac{\partial_t \overline{w_\ast} }{\overline{w_\ast} } \overline{u_\ast w_\ast} V_0'(u_\ast w_\ast)
\right)
\]
by 
\[
\partial_t\overline{u_\ast} 
=\frac{1}{ \overline{w_\ast} } 
\left(
\partial_t (\overline{u_\ast w_\ast})
-
\frac{\partial_t \overline{w_\ast} }{\overline{w_\ast} } 
\overline{u_\ast w_\ast}
\right),
\]
we obtain  
\[
2\mbox{Re}\, I
=
\frac{|a_\ast|^2}{|w_\ast|^2} 
\left(
2\partial_t V_0(u_\ast w_\ast)
-
\left(
2\mbox{Re}\, \frac{\partial_t \overline{w_\ast} }{\overline{w_\ast} }
\right)
\overline{u_\ast w_\ast} V_0'(u_\ast w_\ast)
\right)
\]
by \eqref{Assumption-V0dt} and \eqref{Assumption-zV0}.
Since we have 
\begin{eqnarray*}
\partial_t \left(\frac{1}{ |w_\ast|^{2} }\right)&=&-2\mbox{Re}\, \left(\frac{\partial_t \overline{w_\ast}}{\overline{w_\ast}}\right)
\cdot \frac{1}{|w_\ast|^2}, \\
\partial_t \left(\frac{|a_\ast|^2}{|w_\ast|^2}\right)&=&\frac{n+2}{2|w_\ast|^2} \cdot \partial_t 
\left(|a_\ast|^2\right), \\ 
|a_\ast|^2 \partial_t \left(\frac{1}{ |w_\ast|^{2} }\right)&=&\frac{n}{2|w_\ast|^2}\cdot \partial_t 
\left(|a_\ast|^2\right)
\end{eqnarray*}
by the definition of $w_\ast$, 
we obtain 
\begin{multline}
\label{Proof-Energy-NR-E-10}
2\mbox{Re}\, I
=
2\partial_t 
\left(
\frac{|a_\ast|^2}{|w_\ast|^2} 
V_0(u_\ast w_\ast)
\right)
\\
+
\frac{n\partial_t(|a_\ast|^2)}{2|w_\ast|^2}
\left(
\overline{u_\ast w_\ast} V_0'(u_\ast w_\ast)
-\frac{2(n+2)}{n} V_0(u_\ast w_\ast)
\right).
\end{multline}
So that, we have 
\[
\partial_t e_E^0+\sum_{j=1}^n \partial_{x^j} e_E^j +e_E^{n+1}=0 
\]
by \eqref{Proof-Energy-NR-E} and \eqref{Proof-Energy-NR-E-10}. 
We obtain the required result by the integration for $t$ and $x$.
\end{proof}

Let us consider the case $\theta=\omega^0=\omega^1=0$ in Lemmas 
\ref{Lemma-Conservation-KG} and \ref{Lemma-Conservation-CGL}.
Then the equations \eqref{Eq-Field-2-V-*-V} and \eqref{Eq-Field-2-NR-V-*-V} 
are the Klein-Gordon equation and the Schr\"odinger equation, respectively.
Since $\mbox{Im} (e^{2i(\theta+\omega^1)}/e^{2i\omega^0})=0$, 
we have the strict conservation of the energy 
\[
\int_{\mathbb{R}^n} e^0(t,x)dx=\int_{\mathbb{R}^n} e^0(0,x)dx,
\ \ 
\int_{\mathbb{R}^n} e^0_E(t,x)dx=\int_{\mathbb{R}^n} e^0_E(0,x)dx
\]
when $a_\ast$ is a constant, 
while we see in \eqref{Energy-KG} the dissipative and antidissipative properties by the spatial variance 
when $\partial_t a_\ast>0$ and $\partial_t a_\ast<0$ 
since $e^{n+1}>0$ and $e^{n+1}<0$, respectively.
The effects by the spatial variance also appear in \eqref{Energy-C} and \eqref{Energy-E} 
dependently on the structure of the potential $V_0$.
Let us consider the case $\theta=\omega^0=0$ and $\omega^1=\pi/4$ 
in Lemma \ref{Lemma-Conservation-CGL}.
Then the equation \eqref{Eq-Field-2-NR-V-*-V} with the positive sign is the parabolic equation.
Since $\mbox{Im} (e^{2i(\theta+\omega^1)}/e^{2i\omega^0})=1$, 
the dissipative properties always appear in \eqref{Energy-C} and \eqref{Energy-E} when $V_0=0$.
The effects by the spatial variance also appear in the properties 
dependently on the structure of the potential $V_0$.
Consequently, the spatial variance is one of key factors 
which determine the dissipative and antidissipative properties of the equations.


\newsection{Remarks on Vilenkin's model}
\label{Section-Vilenkin}
Let us consider the model of the birth of the universe by Vilenkin  
and generalize it to the case of general dimensions and complex line elements.
We have derived the line element \eqref{Line-Element-a}, which is from 
\eqref{Line-Element} with \eqref{f}, \eqref{a-h} and \eqref{a}.
We put $\widetilde{f}(z):=h(z^0)+f(r)$.
By a direct calculation, we have 
$g=-c^2 e^{n\widetilde{f}}$, 
$\partial_\alpha (-g)^{1/2}=n(-g)^{1/2} \partial_\alpha \widetilde{f}/2$ 
and 
\[
R=-\frac{n}{c^2} \partial_0^2 h
-\frac{n(n+1)}{4c^2} (\partial_0 h)^2
-\frac{n(n-1) k^2}{q^2} e^{-h}.
\]
By these results, we have 
\[
\int_{\mathcal{M}} R (-g)^{1/2} dz
=
\frac{n-1}{2}
\int_{\mathcal{M}} 
\left\{
\frac{n}{2c^2} (\partial_0 h)^2 -\frac{2n k^2}{q^2} e^{-h}
\right\}
(-g)^{1/2} dz,
\]
where we have used the integration by parts for the term $\partial_0^2h$ in $R$.
So that, we have 
\[
\int_{\mathcal{M}} (R+2\Lambda)(-g)^{1/2} dz
=
\frac{(n-1)\kappa c^3 }{2}
\int_{\mathcal{M}} L(a,\partial_0 a)e^{nf/2} dz,
\]
where we have defined $L(a,\partial_0 a)$ by 
\[
L(a,\partial_0 a):=\frac{a^n}{\kappa c^2}
\left\{
\frac{2n}{c^2}\left(\frac{\partial_0 a}{a}\right)^2-\frac{2nk^2}{a^2q^2}
+\frac{4\Lambda}{n-1}
\right\}.
\]
We regard $L(a,\partial_0 a)$ as the Lagrangian for the variation of $a$.
We define the momentum $p:=\partial L/\partial (\partial_0 a)$ 
and the Hamiltonian $H:=p\partial_0 a-L$.
By the definition of $L(a,\partial_0a)$, we have 
\[
p=\frac{4n a^{n-2} \partial_0 a}{\kappa c^4},
\ \ \ \ 
H=\frac{2n a^n}{\kappa c^2}
\left\{
\frac{1}{c^2}\cdot 
\left(
\frac{\partial_0 a}{a}
\right)^2
+
\frac{k^2}{q^2 a^2}
-
\frac{2\Lambda}{n(n-1)} 
\right\},
\]
by which we also have 
\[H=\frac{2n a^n}{\kappa c^2} 
\left(
\frac{c}{4n a^{n-1}}
\right)^2
\left\{
\kappa^2p^2c^4+\frac{V(a)}{c^4}
\right\},
\]
where we have put a potential 
\[
V(a):=
c^2
(4n a^{n-2})^2
\left(
\frac{k^2}{q^2}-\frac{a^2}{\ell^2}
\right)
\] 
and $\ell :=(n(n-1)/2\Lambda)^{1/2}$.
So that, the solution of the equation of motion $H=0$ is given by  
the scale function
\[
a(z^0)=
\left\{
\begin{array}{ll}
\pm\frac{k\ell}{q} \cosh\left(\frac{cz^0}{\ell}+C\right) & \mbox{if}\ \ k\neq 0, \\
a(0) e^{\pm c z^0/\ell} & \mbox{if}\ \ k= 0
\end{array}
\right.
\]
for some constant $C\in \mathbb{C}$.
Let us consider the case $\pm k/q=1$, $C=0$, $\ell>0$, $\kappa\in \mathbb{R}$, $\Lambda(\neq0)\in \mathbb{R}$  and $z^0=t\in \mathbb{R}$.
Then we have $a(t)= \ell \cosh (ct/\ell)$, and $a$, $L$, $p$, $H$ are real-valued.
We need $a\ge \ell$ for the equation $H=0$ 
since $V>0$ and $H>0$ if $a<\ell$.
This means that  
the universe grows up as the de Sitter spacetime $a(t)=\ell \cosh(ct/\ell)$ for $t\ge0$ in real time $z^0=t$.
To consider the excluded case $a<\ell$, 
let us use the imaginary time $z^0=it$ for $t<0$.
Then we have $a(t)=\ell \cos(ct/\ell)$ for $t<0$.
We need $-\pi \ell/2c <t$ for $a(t)>0$.
Since $V(a)>0$ for $-\pi\ell/2c< t<0$, we have the model that the universe passes through the mountain of the potential $V(a)>0$  by the tunnel effect in imaginary time $z^0=it$ for $-\pi\ell/2c< t<0$, 
then it grows up as the de Sitter spacetime in real time $z^0=t$ for $t\ge0$.
On the other hand, the solution $a(z^0)=a(0)e^{\pm c z^0/\ell}$ for $k=0$ and real time $z^0=t$ does not need the tunnel effect.

%

\newsection{Remarks on the geodesic curves}
\label{Section-Geodesic}
Since the line element \eqref{Tau} is considered in complex coordinates, 
let us consider the geodesic curves derived form it.
We consider the generalized line element of \eqref{Line-Element} given by 
\[
-c^2(d\tau)^2=-c^2(dz^0)^2+g_{jk} dz^j dz^k
\]
for arbitrary complex-valued functions $\{g_{jk}\}_{1\le j,k\le n}$ 
which satisfy the symmetry conditions $g_{jk}=g_{kj}$ for $1\le j,k\le n$.
We put the velocity $v^j:=dz^j/dz^0$ for $1\le j\le n$.
Let $(g^{jk})$ be the inverse matrix of $(g_{jk})$. 
The change of upper and lower indices is done by $g_{jk}$ and $g^{jk}$.
We put $J:=1-v^j v_j/c^2$, and the potential $U=U(z^0,\cdots,z^n)$.
We denote the mass by $m$.
We define the Lagrangian $L$, the momenta $\{p_j\}_{j=1}^n$, 
and the Hamiltonian $H$ by 
\beq
\label{L-H}
L:=-mc^2J^{1/2}-U, \ \ 
p_j:=\frac{\partial L}{\partial v^j}, \ \ 
H:=v^j p_j-L.
\eeq
We put $K:=m^2c^2+p^j p_j$.
By direct calculations, we have 
$d\tau=J^{1/2}dz^0$, 
$\partial J/\partial v^j=-2v_j/c^2$, 
$\partial J/\partial g_{jk}=-v^jv^k/c^2$,
$p^j=m v^jJ^{-1/2}$,
$\partial L/\partial z^\alpha=-\partial U/\partial z^\alpha$,
$\partial L/\partial g_{jk} =J^{1/2} p^j p^k/2m$.
So that, the Euler-Lagrange equation for $L$ is given by 
\begin{eqnarray*}
0 &=& \frac{\partial L}{\partial z^j}
-
\frac{d}{dz^0} \frac{\partial L}{\partial v^j}
+
\frac{\partial L}{\partial g_{\ell m} }\frac{\partial g_{\ell m}}{\partial z^j}
\\
&=&
-
\frac{\partial U}{\partial z^j}
-
\frac{d g_{jk}}{d z^0} p^k
-
g_{jk}\frac{dp^k}{dz^0}
+
\frac{J^{1/2}}{2m} p^\ell p^m \frac{\partial g_{\ell m} }{\partial z^j}.
\end{eqnarray*}
We have 
\beq
\label{K-H}
K^{1/2}=m c J^{-1/2},
\ \ 
H=c K^{1/2}+U
\eeq 
by $p^j p_j=m^2 c^2 (J^{-1}-1)$.
Since we have 
$\partial K/\partial p^j=2p_j$ and $\partial K/\partial g_{jk}=p^j p^k$, 
we have $\partial H/\partial p^j=v_j$,  
$\partial H/\partial g_{jk}=cp^j p^k/2K^{1/2} $ and 
$\partial H/\partial z^\alpha=\partial U/\partial z^\alpha$.
Therefore, we have 
\[
\frac{dH}{dz^0}=
\frac{\partial U}{\partial z^0}-\frac{c}{2K^{1/2} } p^j \frac{\partial g_{jk}}{\partial z^0} p^k
+
\frac{c^2}{2K} p^jp^\ell p^m \frac{\partial g_{\ell m} }{\partial z^j}
=:-H_R,
\]
where we have put the right hand side as $-H_R$.
So that, the Hamiltonian $H$ satisfies the conservation 
\beq
\label{Conservation-H}
H(z^0)+\int_0^{z^0} H_R(w)dw=H(0).
\eeq

Now, we consider the transformation \eqref{z-x} with $\omega^1=\cdots=\omega^n$.
We consider the case $(g_{jk}):=a(z^0)^2 \mbox{\rm diag}(1,\cdots,1)$, 
namely, the case $q=1$ and $k=0$ in \eqref{Line-Element-a}.
Then we have 
\[
J=1-e^{2i(\omega^1-\omega^0)}
\frac{a(e^{i\omega^0} x^0)^2}{c^2} 
\sum_{j=1}^n 
\left(
\frac{dx^j}{dx^0}
\right)^2
\]
by the definition of $J$, and we also have  
\[
e^{i\omega^0} H_R= e^{2i(\omega^1-\omega^0)} \frac{ma}{ J^{1/2} } 
\frac{da}{dx^0} 
\sum_{j=1}^n 
\left(\frac{dx^j}{dx^0}\right)^2 
-\frac{\partial U}{\partial x^0}
\]
by $p^j=me^{i(\omega^1-\omega^0)} J^{-1/2} dx^j/dx^0$, 
\eqref{K-H} and $\partial g_{\ell m}/\partial z^j=0$.
Especially, let us consider the case that 
$\omega^0=0$, $\omega^1=0, \pm \pi/2, \pi$, $m\ge0$.
Let $U$ and $a(>0)$ be real-valued functions.
Let us consider the small velocity such that $J>0$.
Then $J, K, H, H_R$ are real-valued functions.
Moreover, if $-\partial U/\partial x^0\ge0$ and 
\[
\frac{da}{dx^0}
\left\{
\begin{array}{ll}
\ge0 & \mbox{for}\ \omega^1=0, \pi, \\
\le 0 & \mbox{for}\ \omega^1=\pm \frac{\pi}{2},
\end{array}
\right.
\]
then $H_R\ge0$.
So that, the spatial variance $da/dx^0\neq0$ has dissipative 
and antidissipative effects on the 
conservation \eqref{Conservation-H},
while the spatial invariance $da/dx^0=0$ yields the strict conservation 
$H(x^0)=H(0)$ for $x^0\in \mathbb{R}$ 
when $\partial U/\partial x^0=0$.
Namely, the conservation of the Hamiltonian depends on the spatial variance.

Let us consider the above argument for the relativistic velocity $v^\alpha:=dz^\alpha/d\tau$ for $0\le \alpha\le n$ for the general line element \eqref{Tau}.
We put $J:=-v^\alpha g_{\alpha\beta} v^\beta$, and the potential $U=U(z^0,\cdots,z^n)$.
We denote the mass by $m$.
We define the Lagrangian $L$, the momenta $\{p_\alpha\}_{\alpha=0}^n$, 
and the Hamiltonian $H$ by 
\beq
\label{L-H-alpha}
L:=-mc J^{1/2}-U, \ \ 
p_\alpha:=\frac{\partial L}{\partial v^\alpha}, \ \ 
H:=v^\alpha p_\alpha-L.
\eeq
We have 
\[
\frac{\partial L}{\partial z^\alpha} =-\frac{\partial U}{\partial z^\alpha},
\ \ 
p_\alpha=-\frac{mc}{2J^{1/2}} \cdot \frac{\partial J}{\partial v^\alpha},
\ \   
\frac{\partial L}{\partial g_{\alpha\beta} }=-\frac{mc}{2J^{1/2} }\cdot \frac{\partial J}{\partial g_{\alpha\beta}}.
\] 
The Euler-Lagrange equation for $L$ is given by 
\beq
\label{EL-Tao}
\frac{\partial L}{\partial z^\gamma}
-
\frac{d}{d\tau} \left(\frac{\partial L}{\partial v^\gamma}\right)
+
\frac{\partial L}{\partial g_{\alpha\beta} } \cdot \frac{\partial g_{\alpha\beta}}{\partial z^\gamma}
=0.
\eeq
Since $\partial J/\partial v^\gamma =-2v_\gamma$ and 
$\partial J/\partial g_{\alpha\beta}=-v^\alpha v^\beta$,
we have $v_\alpha = J^{1/2}p_\alpha/mc$ and 
$\partial L/\partial g_{\alpha\beta} =mc v^\alpha v^\beta/2J^{1/2}$.
The equation \eqref{EL-Tao} is rewritten as 
\begin{multline}
\label{EL-Tao-Re}
-\frac{\partial U}{\partial z^\gamma}
-\frac{d}{d\tau} 
\left(
\frac{mc}{J^{1/2}}
\right)
\cdot v_\gamma 
-\frac{mc}{J^{1/2}}\cdot \frac{dg_{\gamma \delta}}{d\tau} \cdot v^\delta
\\
-\frac{mc}{J^{1/2}}\cdot  g_{\gamma \delta} \cdot \frac{dv^\delta}{d\tau}
+\frac{mc}{2J^{1/2}}\cdot  
\frac{\partial g_{\alpha\beta}}{\partial z^\gamma}
\cdot  v^\alpha v^\beta=0.
\end{multline}
Multiplying $v^\gamma$ to the both sides in \eqref{EL-Tao-Re} and using the elementary facts 
\[
\frac{dU}{d\tau}=v^\gamma \frac{\partial U}{\partial z^\gamma},
\ \ 
J=-v^\gamma v_\gamma,
\]
\[
v^\gamma g_{\gamma \delta}\frac{d v^\delta}{d\tau}
=\frac{1}{2}
\left(
\frac{d}{d\tau} 
\left( v^\gamma v_\gamma \right) 
-v^\gamma \frac{d g_{\gamma\delta} }{d\tau} v^\delta
\right),
\ \ 
\frac{\partial g_{\alpha\beta} }{\partial z^\gamma} v^\gamma =\frac{d g_{\alpha\beta} }{d\tau},
\]
we have 
\[
\frac{dU}{d\tau}=0.
\]
Since we have $H=U$ by $v^\gamma v_\gamma=-J$ and the definitions of $H$ and $L$, 
the Hamiltonian $H$ is a constant function independent of the proper time $\tau$. 
Namely, we have 
\beq
\label{Conservation-H-Proper}
H(\tau)=H(0).
\eeq

Comparing the conservation laws \eqref{Conservation-H} and \eqref{Conservation-H-Proper}, 
we conclude that the Hamiltonian is strictly conserved 
with respect to the proper time $\tau$ 
independently of the spatial variance, 
while the Hamiltonian is dependent on the spatial variance with respect to the local time $z^0$.

%


%

\end{document}